\documentclass[a4paper, twocolumn]{autart}

\journal{Automatica}

\usepackage{amssymb}
\usepackage{amsmath}
\usepackage{mathptmx}
\usepackage{graphicx}
\usepackage{epsfig}
\usepackage{epstopdf}
\usepackage{caption}
\usepackage{subcaption}

\usepackage{rotate}
\usepackage[usenames]{color}
\usepackage[english]{babel}

\newtheorem{definition}{Definition}
\newtheorem{example}{Example}
\newtheorem{proposition}{Proposition}

\newenvironment{proof}{\begin{trivlist} \item[{ \bf Proof:}] }
{~\hfill$\Box$ \end{trivlist} }

\begin{document}

\begin{frontmatter}
\title{Opinion influence and evolution in social networks: a Markovian agents model}

\author[PoliMI]{Paolo Bolzern}
\ead{paolo.bolzern@polimi.it}
\author[PoliMI,CNR]{Patrizio Colaneri}
\ead{patrizio.colaneri@polimi.it}
\author[UniPV]{Giuseppe De Nicolao}
\ead{giuseppe.denicolao@unipv.it}

\address[PoliMI]{Politecnico di Milano, Dipartimento di Elettronica, Informazione e Bioingegneria\\
Piazza Leonardo da Vinci 32, 20133 Milano, Italy}
\address[CNR]{IEIIT-CNR, Milano, Italy}
\address[UniPV]{Universit\`a di Pavia, Dipartimento di Ingegneria Industriale e dell'Informazione\\
Via Ferrata 5, 27100 Pavia, Italy}


\begin{abstract}
In this paper, the effect on collective opinions of filtering algorithms managed by social network platforms is modeled and investigated. A stochastic multi-agent model for opinion dynamics is proposed, that accounts for a centralized tuning of the strength of interaction between individuals.
The evolution of each individual opinion is described by a Markov chain, whose transition rates are affected by the opinions of the neighbors through influence parameters.
The properties of this model are studied in a general setting as well as in interesting special cases. A general result is that the overall model of the social network behaves like a high-dimensional Markov chain, which is viable to Monte Carlo simulation. Under the assumption of identical agents and unbiased influence, it is shown that the influence intensity affects the variance, but not the expectation, of the number of individuals sharing a certain opinion.
Moreover, a detailed analysis is carried out for the so-called Peer Assembly, which describes the evolution of binary opinions in a completely connected graph of identical agents. It is shown that the Peer Assembly can be lumped into a birth-death chain that can be given a complete analytical characterization.
Both analytical results and simulation experiments are used to highlight the emergence of particular collective behaviours, e.g. consensus and herding, depending on the centralized tuning of the influence parameters.
\end{abstract}

\end{frontmatter}

\section{Introduction}
\label{Introduction}
The pervasive spread of digital social networks in everyday's life of billions of people has marked a dramatic change both in the news propagation and opinions formation. Compared to traditional media (press, TV and radio) operating through broadcast communication, for the first time in history everyone can reach a global audience by means of horizontal diffusion processes.
In spite of their apparent spontaneous and democratic nature, these processes are actually governed by algorithms that filter the potential information presented to each individual users. A notable example is given by Facebook, where the News Feed of each user features only a fraction of the contents posted by her/his contacts. An undisclosed machine learning algorithm ranks the posts accounting for a huge number of factors, including user affinity, user habits and post recentness, see e.g. \cite{Facebook13}. By tuning the algorithm parameters, the social network company exerts {\it de facto} a content-specific control of one-to-one interactions.

A notable experiment, known as {\it Emotional contagion}, was performed by Facebook itself in 2012 \cite{PNAS14}.
A massive subset of unaware users were exposed, during a week, to a change of emotional content in their News Feed and their emotions were then compared to those of a control group. Citing the authors of \cite{PNAS14}: ``When positive expressions were reduced, people produced fewer positive posts and more negative posts; when negative expressions were reduced, the opposite pattern occurred".
This experiment demonstrated that selectively biasing the intensity of certain interactions can affect the emotions of Facebook users. A similar manipulation effect can be conjectured for opinion diffusion across the social network, just by increasing or decreasing the probability that posts in favor of or against a certain opinion will be displayed.

An example of intervention in a political context occurred during the 2010 US Congressional elections. A subset of some 60 million users were encouraged to vote with a message at the top of their News Feed providing indications to local polling places. Moreover, a counter of how many Facebook users had already reported voting was displayed along with profile pictures of friends that had already clicked the {\it I voted} button, \cite{Bond_Nature12}. The results suggested that the social message brought 340,000 additional votes, mostly by social contagion.

In a recent document released by Facebook \cite{Facebook_security17}, the company raised serious concerns about the threat that malicious government or non-state actors may exploit the unprecedented opportunity of global reach to distort political sentiment, e.g. by diffusing fake news or massive disinformation. The development of algorithms curbing the spread of false news in the users' News Feeds is described as one of the main counteractions. On the other hand, the existence of a centralized control of interactions between individuals is a source of worry about the fairness of the confrontation between conflicting opinions.
The invasiveness of Facebook's curatorial function on users' information diets has been highlighted by the so-called {\it Argentina experiment}, whose results have been recently published and discussed by the World Wide Web Foundation \footnote{``The Invisible Curation of Content: Facebook's News Feed and our Information Diets", April 2018.}. In particular, the algorithms were shown to give much different exposure to different stories, making some of them practically invisible to users.
These ethical issues emerged in full evidence during the recent Senate hearing of Facebook's CEO Mark Zuckerberg reporting on data privacy and disinformation on his social network \footnote{``Transcript of Mark Zuckerberg’s Senate hearing", April 10 2018.}.

In view of all this, there is plenty of motivation for developing mathematical models of opinion dynamics in social networks subject to a centralized tuning of the interactions. In recent years, a number of dynamical models for opinion evolution in social networks, based on multi-agent systems, have been proposed. In this framework, each user is represented as an agent and the social connection is modeled by a graph.
A first classification is between real-valued opinions, typically normalized in the interval $[0,1]$, and discrete opinions, described by logical variables taking values in a discrete set. Within the first category, the works \cite{Friedkin15}, \cite{Garulli17}, \cite{Acemoglu13} are worth mentioning. In most of these works, following the original DeGroot \cite{DeGroot74} and Friedkin-Johnsen \cite{Friedkin11} models, the opinion dynamics of each agent is governed by ordinary differential equations influenced by the weighted average of neighbors' opinions. Different types of agents can be considered (e.g. stubborn or influenceable) and the emergence of collective behaviors is investigated. Recent lines of research concern asynchronous interaction \cite{Ravazzi15}, the extension to multi-variate opinions \cite{Parsegov17} and the possibility for a subset of agents to drive the overall network to agree on any desired opinion \cite{Bolouki17}. The interested reader can refer to the recent tutorial \cite{Proskurnikov17} for a fairly complete overview of these classical methods.

In the class of discrete opinion models, Markov chains, see e.g. \cite{Bremaud98} are often adopted to describe the behavior of multiple agents, so that the time evolution of opinions is stochastic in nature. In this class of models, two important contributions are those in \cite{Banisch12} and \cite{Asavathi01}, both defined in a discrete-time framework. Although the Markov chain model of individual opinion dynamics is very appealing in view of its flexibility, the model of the overall social network may soon become analytically untractable as the number of agents grows. In the references above this problem has been overcome by assuming specific interaction mechanisms such that the complete model can be lumped into a lower dimensional one, amenable to an analytical treatment. For instance, in \cite{Banisch12} the interaction occurs only between agents sharing the same opinion, whereas the influence model of \cite{Asavathi01} assumes at each discrete-time step each agent is affected just by a single influencer randomly drawn among its neighbors.
An issue so far not addressed is the presence of a centralized tuning of interactions, which is of particular interest in the management of social platforms. Another interesting class of stochastic Markovian models are those used in the social learning literature to describe the behaviour of economic agents, including speculative bubbles, herding and consensus, see e.g. \cite{Bowden08}, \cite{Jadbabaie12}, \cite{Eksin13}. A distinctive feature of this stream of research is the assumption on the existence of a true ``state of the world" that the agents try to learn through observations and communication. Accordingly, there exist filtering algorithms operated at individual level. These aspects are not dealt with in the present paper, where instead filtering occurs at a centralized level through selective modulation parameters.

Herein, a new continuous-time model is proposed that describes the stochastic evolution of opinion dynamics within a network of Markov agents whose reciprocal influences are modulated in a centralized way. The influence mechanism is essentially borrowed from the multi-agent Markovian networks studied in \cite{Cerotti10} and \cite{BCCG14}. In this paper we assume that the opinion of each agent is modeled by a finite-state continuous-time Markov chain with transition rates dependent on the opinion of the neighboring agents. To be precise, we assume a linear emulative influence mechanism where the transition rates towards a certain opinion linearly depend on the fraction of neighbors sharing that particular opinion. Notably, we provide an exact analysis of the stochastic model, rather than a mean field approximation, like that in \cite{BCCG14} where the transition rates were assumed to be influenced by the probability of neighbors' opinions. Similar mean-field approximations have also been used in the context of SIS (Susceptible-Infected-Susceptible) epidemic models, see e.g. \cite{VanMieghem11}. As a further difference with respect to our work, the model of \cite{VanMieghem11} deals with two-state Markov chains where only the transition rate to the ill state is affected by neighbors. Our model covers a much wider range of scenarios, the number of states being arbitrary and all transition rates influenceable by interaction.

In order to describe a controlled social network, tunable influence parameters are introduced. Each influence parameter can enhance or curb the spread of a specific opinion. The use of identical influence parameters accelerates opinion dynamics in an unbiased way. Conversely, a bias can be induced by the adoption of unbalanced parameters, up to the limit case of unilateral promotion of a single opinion.
The model is believed to be simple yet flexible enough to account for a variety of possible situations, regarding both the model of the single agent and the centralized manipulation of opinions diffusion.

The purpose of this paper is twofold: first of all, the new general framework is proposed and its basic properties are established. In particular, it is shown that the overall social network is a high dimensional Markov chain (hereafter called Master Markov model). Second, some analytical properties of the model are investigated under particular assumptions.
For identical agents and unbiased influence, we prove that, independently of the network topology, the probability distribution of the opinion of any agent converges towards a common consensus distribution, which is notably not affected by the value of the influence parameter.
In the special case of identical agents, binary opinions and complete network topology (dubbed Peer Assembly in the sequel), the Master Markov model can be lumped into a birth-death Markov chain, thus dramatically simplifying the analysis. In particular, under the additional assumption of unbiased influence, explicit expressions for expectation and variance of the fraction of agents in one opinion are provided.
The emergence of some interesting collective behaviors is also discussed. For instance, an increase of the influence strength parameter, though ineffective with respect to the average fraction of agents in a certain opinion, does increase the volatility of this fraction. Moreover, very large values of the influence parameters eventually trigger a herding phenomenon, giving rise to massive opinion waves.

The paper is organized as follows. After having introduced some notation (Section II), the model of interacting Markovian agents is presented in Section III. Section IV is devoted to show that the model is amenable to marginalization when the agents share the same model and the influence is unbiased. Section V presents a detailed analysis of the so-called Peer Assembly model in both cases of unbiased and biased influence. Finally, the case of general network topologies is illustrated in Section VI by means of some simulation examples. The paper ends with some concluding remarks (Section VII).

\section{Notation}
\label{Notation}

A square matrix $A=[a_{ij}] $ is said to be {\em Metzler} if its off-diagonal entries are nonnegative, namely $a_{ij}\ge 0$ for every $i\ne j$.

An $n \times n$ Metzler matrix $A$, with $n > 1$,  is  {\em reducible} if  there exists a permutation matrix $P$ such that
$$P' A P = \left[\begin{matrix}A_{11} & A_{12}\cr 0 & A_{22}\end{matrix}\right],$$
where $A_{11}$ is  a $k\times k$ matrix, $1 \le k \le n-1$. A Metzler matrix
that is not reducible is called {\em irreducible}, see Chapter 2 of \cite{Berman94}.

The vector ${\bf 1}_n$ is the $n$-dimensional column vector with all entries equal to $1$. The suffix $n$ will be omitted when the vector size is clear from the context. The symbol $\otimes$ stands for the Kronecker product. Given a discrete set $\mathcal N$, the symbol $|\mathcal N|$ denotes its cardinality.

Let $\sigma(A)$ denote the spectrum of a Metzler matrix $A$. It is known that
$\max \{{\rm Re}(\lambda): \lambda\in \sigma(A)\}$ is always an eigenvalue of $A$, called the  {\em Perron-Frobenius eigenvalue} and denoted by $\lambda_F$, and that the corresponding eigenspace, when $A$ is irreducible, is generated by a positive eigenvector, $v_F$, with ${\bf 1}'v_F=1$, called {\em Frobenius eigenvector} \cite{Berman94}, Chapter 2. The set of probability vectors, i.e. vectors with nonnegative entries that sum up to 1 will be indicated as $\mathcal P$.

The symbols $E[v]$ and $Var[v]$ denote the expectation and variance of the random variable $v$. Given a random event $\mathcal{E}$, $\mathcal{I}_{\mathcal E}$ represents the indicator function of the event, namely $\mathcal{I}_{\mathcal E}=1$ if $\mathcal{E}$ occurs, $\mathcal{I}_{\mathcal E}=0$ otherwise. ${\rm Pr}\{\mathcal{E}\}$ will be used to denote the probability of the event $\mathcal{E}$ and ${\rm Pr}\{\mathcal{E}|\mathcal{F}\}$ is the conditional probability of $\mathcal{E}$ given the event $\mathcal{F}$.

\section{A model of interacting Markovian agents}
\label{Model}
The interaction in a social network can be described through an undirected graph ${\mathcal G}=({\mathcal N},{\mathcal E})$, with a finite set of nodes ${\mathcal N}=\{1,2,\ldots,N\}$ representing the individuals (agents) and edges ${\mathcal E} \subseteq {\mathcal N} \times {\mathcal N}$ associated to reciprocal influences. An edge connecting node $r$ to node $s$ means that the two agents influence each other. The set of neighbors of node $r$ will be denoted by ${\mathcal N}^{[r]} = \{s \in {\mathcal N} : (r,s) \in {\mathcal E}\}$.

\subsection{Stand-alone model}
\label{StandaloneModel}
At a given time, each individual in the network has an opinion belonging to a finite set ${\mathcal M}=\{1,2,\ldots,M\}$.
Let $\sigma^{[r]}(t) \in {\mathcal M}$, $r \in {\mathcal N}$, $t \in \mathbb{R}$, be the state of the $r$-th individual at time $t$.
It is assumed that, when isolated, the change in the opinion of the individual $r$ evolves according to a finite-state continuous-time time-homogeneous Markov chain with transition rate matrix $Q^{[r]} \in \mathbb{R}^{M\times M}$. The entries of matrix $Q^{[r]}$ are denoted as $q_{ij}^{[r]}$ and represent the transition rates between opinions. To be more precise, when $i \ne j$, it results that
\begin{equation}
\label{poisson}
\Pr\{\sigma^{[r]}(t+dt)=j|\sigma^{[r]}(t)=i\}=q_{ij}^{[r]}dt+o(dt)
\end{equation}
with $q_{ij} \ge 0$. The diagonal entries of $Q^{[r]}$ are defined as
$$
q_{ii}^{[r]}=-\sum_{j=1,j \ne i}^M q_{ij}^{[r]}
$$
so that $Q^{[r]}$ is a Metzler matrix satisfying $Q^{[r]}{\bf 1}=0$.

As is well known, a Markov chain can be described by the probability of being in a certain state $i$ at time $t$.
Precisely, define $\pi_i^{[r]}(t)=\Pr\{\sigma^{[r]}(t)=i\}$ and let
$$
\pi^{[r]}(t)=[\begin{array}{ccc}\pi_{1}^{[r]}(t)&\dots&\pi_{M}^{[r]}(t)\end{array}]^\prime
$$
be the probability distribution vector at time $t$.
It is well known that, given an initial probability distribution
$$
\pi^{[r]}(0)=[\begin{array}{ccc}\pi_{01}^{[r]}&\dots&\pi_{0M}^{[r]}\end{array}]^\prime
$$
where $\pi_{0i}^{[r]}:=\Pr\{\sigma^{[r]}(0)=i\}$, the time evolution of the probability distribution
$\pi^{[r]}(t)$ from the zero initial time instant obeys the differential equation
\begin{equation}
\label{eq_sa}
\dot\pi^{[r]}(t) = Q^{[r]\prime} \pi^{[r]}(t)
\end{equation}
If the transition rate matrix is irreducible (see, e.g. Chapter 3 of \cite{Bremaud98}), then $\sigma^{[r]}(\cdot)$ is ergodic. Therefore, for any initial $\pi^{[r]}(0) \in \mathcal P$, $\pi^{[r]}(t)$ converges, as $t \to \infty$, to a strictly positive stationary probability vector $\bar\pi^{[r]} \in \mathcal P$ which is the unique unit-sum Frobenius left eigenvector of $Q^{[r]}$ associated with the Perron-Frobenius null eigenvalue, see Chapter 2 of \cite{Berman94}. In the sequel, it is assumed that the matrix $Q^{[r]}$ is irreducible, $\forall r \in {\mathcal N}$. Roughly speaking, this assumption implies that, on any finite time interval, each agent can move from any opinon to any other with nonzero probability.

\subsection{Interaction model}
\label{InteractionModel}
We suppose that the interaction within the social network affects the individual behavior making it different from the stand-alone model of Section \ref{StandaloneModel}. Hence, at time $t$ the transition rate matrix $Q^{[r]}$ is replaced by
\begin{equation}
\label{eq.inter}
\tilde Q^{[r]}\left(\{\sigma^{[k]}(t), k \in {\mathcal N}^{[r]} \}\right)
\end{equation}
i.e. the transition rates are influenced by the neighbors' states.
Therefore, the single agent, conditional on the states of its neighbors, behaves as an inhomogeneous Markov chain. This model will be dubbed as the {\it atomic interaction model}.

In the following, we will focus on the special case where
\begin{equation}
\label{eq.linemu1}
\tilde Q^{[r]} =  Q^{[r]} + A^{[r]}(t)
\end{equation}
where, for $i \ne j$,
\begin{eqnarray}
\label{eq.linemu2}
a_{ij}^{[r]}(t) &=& \frac{\lambda_j}{|{\mathcal N}^{[r]}|} \sum_{k \in {\mathcal N}^{[r]}} \mathcal{I}_{\sigma^{[k]}(t)=j}\\
\label{eq.linemu3}
\mathcal{I}_{\sigma^{[k]}(t)=j} &=& \left\{\begin{array}{ll} 1, & \sigma^{[k]}(t)=j\\ 0, & {\rm otherwise}\end{array}\right.
\end{eqnarray}
and the elements $a_{ii}^{[r]}(t)$ are such that $A^{[r]}(t){\bf 1} = 0$. In this case, the model describes an emulative behavior, namely the instantaneous transition rates to opinion $j$ undergo an increase which is proportional to the number of neighbors that share opinion $j$. In the sequel, such an interaction model will be dubbed as {\it linear emulative}. The parameter $\lambda_j$ reflects the influence strength intensity, that might be different for different opinions. The term {\em unbiased influence} will be used to indicate the case $\lambda_i=\lambda$, $\forall i \in \mathcal M$.
In other cases, the parameter $\lambda_j$ may represent a tuning knob used by the network manager to strengthen (or weaken) the circulation of opinion $j$ within the network, thus exerting a {\em biased influence}. For instance, in Facebook this tuning may result from interventions deployed to curb the spread of fake news.

\begin{rem}
Note that the atomic interaction model is characterized by transition rate matrices for each agent which depend on the state of the neighbors. As such, these are random matrices and the overall network is described by a state-dependent Markov process. This modeling approach is usually referred to as {\it Stochastic Automata Network}, see e.g \cite{Stewart94} and related literature.
In general, the rigorous probabilistic analysis of both the transient and steady-state properties of the model is rather difficult. However, if one is interested in an approximate evaluation of the opinion dynamics, the model is amenable for Monte Carlo distributed simulation, with time-discretization. Precisely, for each agent, the transition probabilities in the interval $[t, t+dt]$ are computed given the state at time $t$ of its neighbors. At time $t+dt$ the transition probabilities are updated and the simulation is iterated. In alternative to discretization, one can resort to a Gillespie-type simulation algorithm, see Chapter 4 of \cite{Gillespie92}.
\end{rem}

\subsection{Master Markov model}
\label{MasterMarkovModel}
Based on the atomic interaction model previously described, a model for the entire network of agents is now introduced. First, consider the case of a network with $N$ non-interacting agents. Each agent $r$ is described by a Markov chain with transition rate matrix $Q^{[r]}$ and all chains are statistically independent. Let
$$
\Sigma(t)=\begin{bmatrix}\sigma^{[1]}(t) & \sigma^{[2]}(t) & \cdots & \sigma^{[N]}(t)\end{bmatrix}^\prime \in \mathcal M^N
$$
be the state of the network and $\pi(t)=\otimes_{r=1}^N\pi^{[r]}(t)$ denote the probability distribution of $\Sigma(t)$. The entries of vector $\pi(t)$ represent the probability at time $t$ of a given configuration of the opinions in the network.
Simple computation shows that, in view of independence, $\pi(t)$ satisfies the differential equation
$$
\dot \pi(t) = Q_0^\prime \pi(t)
$$
with
$$
Q_0 = \sum_{r=1}^N I_{M^{r-1}}\otimes Q^{[r]}\otimes I_{M^{N-r}}
$$
Hence, the network itself is a homogeneous Markov chain whose state is the Cartesian product of the agents' states and has dimension $M^N$. Since the agents' states evolve as independent ergodic processes, the joint process $\Sigma(\cdot)$ is ergodic as well. This in turn implies that $Q_0$ is an irreducible matrix, see e.g. \cite{Bremaud98}, Chapter 3.

Now suppose that the network is formed by interacting agents, with a mutual influence mechanism described as in Section \ref{InteractionModel} by means of eq. \eqref{eq.inter}.

\begin{proposition}
Let $\Sigma(t)\in \mathcal M^N$ be the state at time $t$ of a network composed by $N$ Markovian agents interacting according to eq. \eqref{eq.inter}. Then the process $\Sigma(t)$ is a time-homogeneous Markov chain.
\end{proposition}

Indeed, given the knowledge of $\Sigma(t)$ at a certain $t$, the probability distribution of $\Sigma(t+dt)$ conditional on the past up to time $t$ does coincide with the probability distribution of $\Sigma(t+dt)$ given $\Sigma(t)$, so that the Markov property holds. Moreover, the transition rates do not depend on time.

When the interaction model is linear emulative, see eqs. \eqref{eq.linemu1}-\eqref{eq.linemu3}, the dynamics of the probability distribution $\pi(t)$ is described by
\begin{equation}
\label{eq.master}
\dot \pi(t) = (Q_0+A_0)^\prime \pi(t)
\end{equation}
with a suitable definition of the matrix $A_0$, where $A_0{\bf 1} = 0$. This matrix takes into account both the network topology and the fact that simultaneous opinion jumps of two or more agents are not allowed (the probability of simultaneous jumps in $dt$ is of order $o(dt)$). Moreover, if $Q_0$ is irreducible, so is the matrix $Q_0+A_0$. This implies asymptotic stationarity and ergodicity of the Markovian process $\Sigma(t)$. In particular, for any $\pi(0)$,
$$
\lim_{t \to \infty} \pi(t) = \bar \pi
$$
where $\bar \pi$ is the unique steady state probability distribution.

\begin{example}
As an example, consider a network composed by three agents ($N=3$) jumping between two opinions ($M=2$). Moreover, the interaction graph is complete, i.e. $\mathcal N^{[r]}= \mathcal{N}/{r}$, $\forall r$, and the interaction is linear emulative with parameters $\lambda_1$, $\lambda_2$. In this case the matrix $A_0$ is given by
\begin{equation*}
\left[\begin{smallmatrix}
0 & 0 & 0 & 0 & 0 & 0 & 0 & 0 \\
\lambda_1 & -(\lambda_1+\lambda_2) & 0 & \lambda_2/2 & 0 & \lambda_2/2 & 0 & 0\\
\lambda_1 & 0 & -(\lambda_1+\lambda_2) & \lambda_2/2 & 0 & 0 & \lambda_2/2 & 0\\
0 & \lambda_1/2 & \lambda_1/2 & -(\lambda_1+\lambda_2) & 0 & 0 & 0 & \lambda_2\\
\lambda_1 & 0 & 0 & 0 & -(\lambda_1+\lambda_2) & \lambda_2/2 & \lambda_2/2 & 0 \\
0 & \lambda_1/2 & 0 & 0 & \lambda_1/2 & -(\lambda_1+\lambda_2) & 0 & \lambda_2\\
0 & 0 & \lambda_1/2 & 0 & \lambda_1/2 & 0 & -(\lambda_1+\lambda_2) & \lambda_2\\
0 & 0 & 0 & 0 & 0 & 0 & 0 & 0
\end{smallmatrix}\right]
\end{equation*}
Note that the first and the last rows are zero because they correspond to transitions starting from either $\Sigma(t)=[1 \,\, 1 \,\, 1]^\prime$ or $\Sigma(t)=[2 \,\, 2 \,\, 2]^\prime$, none of which has an emulative incentive. The off-diagonal nonzero entries of the second row corresponds to transitions starting from $\Sigma(t)=[1 \,\, 1 \,\, 2]^\prime$. The emulative increment of the transition probability to $\Sigma(t+dt)=[1 \,\, 1 \,\, 1]^\prime$ is $\lambda_1 dt$ since the third agent emulates the opinion $\sigma=1$ of the first two. The emulative increment of the transition probability to $\Sigma(t+dt)=[1 \,\, 2 \,\, 2]^\prime$ is $\lambda_2 dt/2$ since the second agent changes its opinion to $\sigma=2$ in accordance with the third agent only. All other off-diagonal entries of matrix $A_0$  can be explained with similar arguments.
\end{example}

The state of the master Markov model is of dimension $M^N$, which soon renders its use prohibitive for the evaluation of the probability distribution as the number $N$ of nodes increases. Nevertheless, in view of ergodicity, stochastic simulation is a viable way to assess the stationary properties of the process. In fact, by the ergodic law of large numbers, any statistics of interest, e.g. the steady state probability distribution $\bar \pi$ or the steady state marginal  distribution of agent $r$, $\bar \pi^{[r]}$, can be estimated from a sufficiently long sequence of simulated samples. Details on a possible simulation algorithm are given in \cite{Gillespie92}.

By observing that $\bar \pi$ if the left Frobenius eigenvector of the large dimension matrix $Q_0+A_0$, one may devise alternative distributed algorithms that exploit the sparsity of the matrix, similarly to what has been done for PageRank computation in \cite{Ishii10}.

\section{Marginalization of the linear emulative model}
In the special case of the linear emulative model with identical agents ($Q^{[r]}=Q$) and unbiased influence, i.e. with $\lambda_i=\lambda, \forall i \in \mathcal M$, it is possible to work out a $NM$-dimensional dynamic linear system describing the propagation of the marginal distributions $\pi^{[r]}_j(t)$ of all agents.
Remarkably, irrespective of the network topology, the agents reach asymptotically a probabilistic consensus $\tilde \pi$, coincident with the stand-alone steady state distribution
According to our notation, denote with $\pi^{[r]}_j(t)$ the probability that agent $r$ has opinion $j$ at time $t$.

\begin{thm}
\label{th.marginalization}
For the unbiased influence linear emulative model with an arbitrary network topology, it holds that $\pi^{[r]}_j(t), t\ge 0$, is the solution of the following linear differential equation:
\begin{equation}
  \dot \pi^{[r]}_j(t) = \sum_{i=1}^M q_{ij} \pi^{[r]}_i(t)
  + \lambda \left( \frac{\displaystyle\sum_{k \in {\mathcal N}^{[r]}} \pi^{[k]}_j(t)}{|{\mathcal N}^{[r]}|}  - \pi^{[r]}_j(t)  \right)
  \label{eq.unbmar}
\end{equation}
\end{thm}

\begin{proof}
For simplicity, we use the abridged notation $\mathcal{I}_j^r(t)$ for the indicator function $\mathcal{I}_{\sigma^{[r]}(t)=j}$ in \eqref{eq.linemu3}. Moreover define
$$
\Delta_i(r,t)=\frac{\lambda_i}{|{\mathcal N}^{[r]}|} \sum_{k \in \mathcal N^{[r]}} \mathcal{I}_i^k(t)
$$
Then, considering the agent $r$ and recalling $\sum_{i \ne j}q_{ji}=-q_{jj}$, we have that, $\forall j \in \mathcal M$,
\begin{eqnarray*}
 &&E[\mathcal{I}_j^r(t+dt)| \Sigma(t)] = \mathcal{I}_j^r(t)\left[1 - dt \sum_{i \ne j}\left( q_{ji} +
  \Delta_i(r,t)\right) \right] \\
  &&+ \left(1-\mathcal{I}_j^r(t)\right)dt\left( \sum_{i \ne j}q_{ij}\mathcal{I}_i^r(t) +
  \Delta_j(r,t)\right) \\
  &=& \mathcal{I}_j^r(t) - \mathcal{I}_j^r(t)dt \left( -q_{jj} + \sum_{i \ne j} q_{ij} \mathcal{I}_i^r(t) +
  \sum_{i} \Delta_i(r,t) \right) \\
  &&+ dt \left( \sum_{i \ne j} q_{ij} \mathcal{I}_i^r(t) + \Delta_j(r,t) \right)
\end{eqnarray*}
By observing that $\sum_{i}\mathcal{I}_i^k(t)=1, \forall k$, $\mathcal{I}_j^r(t)\mathcal{I}_i^r(t)=0, \forall r$, for $i \ne j$ and $\lambda_i=\lambda, \forall i$, we obtain that
\begin{eqnarray*}
 &&E[\mathcal{I}_j^r(t+dt)| \Sigma(t)] =
  \mathcal{I}_j^r(t) - \lambda \mathcal{I}_j^r(t)dt  \\
  &&\hspace{1cm}+ dt \left( \sum_{i} q_{ij} \mathcal{I}_i^r(t) + \frac{\lambda}{|{\mathcal N}^{[r]}|} \sum_{k \in \mathcal N^{[r]}} \mathcal{I}_j^k(t) \right)
\end{eqnarray*}
By taking the expectation and noticing that $E[\mathcal{I}_j^r(t)]=\pi_j^{[r]}(t)$, the differential equation \eqref{eq.unbmar} directly follows.
\end{proof}

Note that, if the initial distribution probabilities of all agents' opinions are equal ($\pi^{[r]}_j(0)$ does not depend on $r$ for all $j$), then the term between brackets in \eqref{eq.unbmar} is null for all times and \eqref{eq.unbmar} boils down to the stand-alone time-evolution described by \eqref{eq_sa}. The next theorem shows that the steady-state probabilities of the stand-alone case are recovered asymptotically even when the initial probability distributions of the agents' opinions are different.

\begin{thm}
\label{th.gattosenzatopo}
For the unbiased influence linear emulative model with an arbitrary network topology, the probability distributions of the agents' opinions $\pi^{[r]}(t)$ reach asymptotically a steady-state consensus represented by the steady-state stand-alone distribution $\tilde \pi$, independently of the initial probability distribution.
\end{thm}

\begin{proof}
Since the matrix $Q_0+A_0$ is irreducible, the master Markov model is ergodic. Then, the solution $\pi(t)$ of \eqref{eq.master} asymptotically tends to a unique steady state probability distribution $\bar\pi$. In turn, the marginal distribution $\pi^{[r]}(t)$ of each agent converges to a unique steady-state value $\tilde \pi^{[r]}$ for any initial condition. This means that $\tilde \pi_j^{[r]}$, $j \in \mathcal M$, $r \in \mathcal N$ specify an equilibrium point of the system of differential equations \eqref{eq.unbmar}. On the other hand, it is easy to verify that the steady-state stand-alone distribution $\tilde \pi$ is an equilibrium point of \eqref{eq.unbmar}. This concludes the proof.
\end{proof}

\begin{rem}
The marginalized model \eqref{eq.unbmar} can be interpreted as an extension to continuous-time and generalization of the so-called {\it Influence model} presented in \cite{Asavathi01}. As a main difference, we assume here that each agent is simultaneously influenced by all its neighbors, whereas in \cite{Asavathi01} the interaction mechanism is based on the selection of a single randomly chosen influencing agent at each time step. Thanks to this last assumption, the marginalization of the underlying master model is always guaranteed. However, in that model it seems impossible to treat the case of biased influence.
\end{rem}

\section{A special case: the Peer Assembly}
\label{peer}
The Master Markov Model of the previous section, although of high dimension ($M^N$), can be studied with standard tools of finite-state Markov chains. However the analysis of the effects of the interaction parameters $\lambda_j$ may soon become intractable as the state dimension grows, unless the analysis is restricted to special cases. In this section, we consider a social network composed by $N$ identical individuals with binary opinions ($M=2$), interconnected by a complete graph (each individual communicates with all others) and sharing the same stand-alone irreducible $2\times2$ transition rate matrix $Q=[q_{ij}]$ (the attitude of agents to opinion changes when isolated is identical for all individuals). In view of irreducibility, $q_{ij} \ne 0, i \ne j$. It is also assumed that the initial opinions of each agent are independent and identically distributed random variables. Moreover linear emulative interaction is assumed, as described by eqs. \eqref{eq.linemu1}-\eqref{eq.linemu3}, namely the increment of the transition rate of each agent towards a certain opinion is linearly influenced by the number of neighbors that share that opinion. This model will be referred to as the {\it Peer Assembly} (PA) model.

Due the intrinsic indistinguishability of the agents, this model can be lumped into a birth-death Markov process, see Chapter 6 of \cite{Kemeny76} for a discussion on lumpability of Markov models and Section 7.4 of \cite{Cassandras08} for classical birth-death Markov processes. The state of this process is the number $n_1(t)\in \{0,1,\ldots,N\}$ of individuals having opinion 1 at time $t$, namely $n_1(t)=\sum_{r=1}^{N} \mathcal{I}_{\sigma^{[r]}(t)=1}$. Note that, since the opinion is binary, the number $n_2(t)$ of individuals having opinion 2 at time $t$ is $n_2(t)=N-n_1(t)$.
Obviously, the cardinality of the state space of the PA model is $N+1$, dramatically reducing the original cardinality $2^N$ of the Master Markov model. The PA model is well suited to describe the time evolution of the opinion share, a notable example being the case of election polls.

\begin{definition}
  A birth-death chain is a continuous-time Markov chain with tridiagonal transition rate matrix $\Psi=[\psi_{i,j}]$. The birth rates are the upper diagonal entries $\mu_j = \psi_{j,j+1}, j \ge 1$, while the death rates are the lower diagonal entries $\nu_j = \psi_{j+2,j+1}, j \ge 0$. The (nonpositive) diagonal entries are such that the entries of each row sum to zero.
\end{definition}

In our case, $\mu_j dt={\rm Pr}\{n_1(t+dt)=j|n_1(t)=j-1\}$ is the probability that the number of the individuals with opinion 1 increases from $j-1$ to $j$ in the interval $dt$. Likewise,
$\nu_j dt={\rm Pr}\{n_1(t+dt)=j|n_1(t)=j+1\}$ is the probability that the number of the individuals with opinion 1 decreases from $j+1$ to $j$ in the interval $dt$.

\begin{proposition}
  For the PA model, the number $n_1(t)$ of individuals sharing opinion 1 at time $t$ evolves as a finite-state irreducible birth-death chain with
\begin{eqnarray}
\label{eq.mu}
  \mu_j &=& \left(q_{21} + \lambda_1 \frac{j-1}{N-1} \right)(N-j+1), \quad 1 \le j \le N \\
\label{eq.nu}
  \nu_j &=& \left(q_{12} + \lambda_2 \frac{N-(j+1)}{N-1} \right)(j+1), \quad 0 \le j \le N-1
\end{eqnarray}
\end{proposition}

\begin{proof}
Assume that, at time $t$, $n_1(t)=j-1$. In order to have $n_1(t+dt)=j$ it is necessary that one out of the $N-j+1$ agents having opinion 2 switches to opinion 1 in the time interval $dt$. For each of those agents, the probability of switching is $(q_{21} + \lambda_1 (j-1)/(N-1))dt$. This proves formula \eqref{eq.mu}. The formula \eqref{eq.mu} for $\nu_j$ follows from similar arguments.
\end{proof}

Let $p(t)=[p_0(t) \,\, p_1(t) \,\, \cdots \,\, p_N(t)]^\prime$ denote the probability distribution of $n_1(t)$, i.e.
$p_i(t)={\rm Pr}\{n_1(t)=i\}$. It is well known that
\begin{equation}
\label{eq.trans}
\dot p(t) = \Psi' p(t)
\end{equation}
It is interesting to obtain the steady state distribution $\bar p$, which is reached asymptotically in time, thanks to ergodicity (implied by irreducibility of the Markov chain). To this purpose, it is known, see e.g. \cite{Cassandras08}, that
\begin{eqnarray}
\label{eq.pi}
  \bar p_i &=& \bar p_0 \left( \frac{\mu_1\mu_2 \cdots \mu_i}{\nu_0\nu_1 \cdots \nu_{i-1}} \right), \quad i=1,2,\ldots N \\
\label{eq.p0}
  \bar p_0 &=& \left( 1+ \sum_{i=1}^{N}\frac{\mu_1\mu_2 \cdots \mu_i}{\nu_0\nu_1 \cdots \nu_{i-1}}\right)^{-1}
\end{eqnarray}
From the knowledge of the steady state distribution $\bar p_i, i=0, \ldots, N$, it is possible to compute all relevant moments. For instance, the stationary expected value and variance of $n_1/N$, i.e. the fraction of individuals with opinion 1, are respectively given by
\begin{eqnarray}
\label{eq.mean}
E[n_1/N] &=& \left(\sum_{i=0}^N i \bar p_i\right)/N \\
\label{eq.variance}
{\rm Var}[n_1/N] &=& \frac{1}{N^2} \left( E[n_1^2] - E[n_1]^2\right)
\end{eqnarray}
where $E[n_1^2]=\sum_{i=0}^N i^2 \bar p_i$. These indices are particularly interesting for opinion polls. More precisely, assume that an opinion poll aimed at estimating $n_1/N$ is carried out at a given time $t$ interviewing a random subset of the $N$ agents. The confidence interval of this poll should account for two sources of variability, one related to finite sampling and the other depending on ${\rm Var}[n_1/N]$.

\begin{example}
\label{ex.ui3}
Consider the toy example in which the population consists of $N=3$ agents. By plugging \eqref{eq.mu},\eqref{eq.nu},\eqref{eq.pi},\eqref{eq.p0} into \eqref{eq.mean}, a direct computation shows that
$$
E[n_1/N] = \frac{q_{21}\left(\phi(q,\lambda_1) + q_{12}(\lambda_2-\lambda_1)\right)}{q \phi(q,\lambda_1) +
q_{12} \left(3q(\lambda_2-\lambda_1)+ (\lambda_2^2-\lambda_1^2) \right)}
$$
where
$$
q = q_{12}+ q_{21}, \quad \phi(q,\lambda_1) = 2q^2 + 3q \lambda_1 + \lambda_1^2
$$
When the agents are isolated ($\lambda_1=\lambda_2=0$) it turns out that $E[n_1/N]=q_{21}/q$, as expected since it corresponds of the probability that a single agent has opinion 1. When the agents interact and the influence strength intensity is unbiased ($\lambda_1=\lambda_2=\lambda$) the expected value coincides with that of the isolated case, in accordance with Theorem \ref{th.gattosenzatopo}. This means that the expected percentage of opinions is not affected by the interaction. In the next subsection we will investigate on the probability distribution of the opinions for an arbitrary size of the peer assembly.
\end{example}

\subsection{Unbiased influence}
Consider the model of the Peer Assembly with equal influence intensity, i.e. $\lambda_1=\lambda_2=\lambda$, that will be hereafter referred to as {\it Unbiased Influence Peer Assembly} (UIPA).

When the individuals are not interacting ($\lambda=0$), the network consists of $N$ identical independent Markovian agents. The probability distribution of each agent obeys the differential equation \eqref{eq_sa} with
$$
Q^{[r]}=Q=\begin{bmatrix}
            -q_{12} & q_{12} \\
            q_{21} & -q_{21}
          \end{bmatrix}, \forall r
$$
The steady state probability distribution common to all agents is therefore
$$
\tilde \pi = \begin{bmatrix}
               \tilde \pi_1 \\
               \tilde \pi_2
             \end{bmatrix}
             = \begin{bmatrix}
               q_{21} \\
               q_{12}
             \end{bmatrix}\frac{1}{q_{12}+q_{21}}
$$
Since this distribution is Bernoulli-like, its variance is easily computed as $\tilde \sigma^2 = \tilde \pi_1 (1-\tilde \pi_1)$.
The steady-state mean and variance of $n_1/N$ are therefore $E[n_1/N]=\tilde \pi_1$ and, in view of independence, ${\rm Var}[n_1/N]=\tilde \sigma^2/N$.

Let us now extend the analysis to describe also interacting agents ($\lambda\ge0$) both in the transient and in steady-state.

\medskip
\noindent {\bf Transient analysis}
\medskip

For what concerns the time evolution of the probability distribution of each agent, the following result directly follows from application of Theorem \ref{th.marginalization} and the assumption that the agents are indistinguishable at the initial time.
\begin{proposition}
\label{th.UIPAmean}
  For the UIPA model, it holds that, for each agent $r$,
\begin{eqnarray*}
\nonumber
  \dot \pi^{[r]}_1(t) &=& -(\lambda+q_{12}+q_{21})\pi^{[r]}_1(t)+\frac{\lambda}{N-1} \sum_{k \ne r} \pi^{[k]}_1(t) + q_{21}\\
  &=& -(q_{12}+q_{21})\tilde\pi_1(t) + q_{21}
  \label{eq.UIPA1}
\end{eqnarray*}
\end{proposition}

As already observed after Theorem \ref{th.marginalization}, this marginal probability evolution coincides with the stand-alone one, irrespective of the value of $\lambda$. In turn, we have that $E[n_1(t)]=N \tilde \pi_1(t)$.

Now, we consider the time evolution of the joint distribution of the opinions of a generic couple of agents, say $r$ and $s$. Note that, in view of the PA assumption, all couples are equivalent. It will be shown that, differently from the univariate marginal probability of a single agent, this marginal joint distribution is affected by the value of the influence intensity $\lambda$.
For short, denote by $\pi_{ij}^{[rs]}(t)=E[\mathcal{I}_{\sigma^{[r]}(t)=i}\mathcal{I}_{\sigma^{[s]}(t)=j}]$ the probability that agent $r$ has opinion $i$ and agent $s$ has opinion $j$ at time $t$. The following result, whose proof is given in the Appendix, holds.
\begin{thm}
\label{th.UIPAjoint}
Consider the UIPA model. Then for each couple of agents $r$, $s$, it results that
\begin{eqnarray}
\nonumber
  \begin{bmatrix}
    \dot \pi_{11}^{[rs]}(t) \\
    \dot \pi_{22}^{[rs]}(t)
  \end{bmatrix} &=& -\left( 2\begin{bmatrix}
                    q_{12} & 0 \\
                    0 & q_{21}
                  \end{bmatrix} + \begin{bmatrix}
                                                q_{21}+\frac{\lambda}{N-1}  \\
                                                q_{12}+\frac{\lambda}{N-1}
                                              \end{bmatrix} \begin{bmatrix}
                                                              1 & 1
                                                            \end{bmatrix} \right)
  \begin{bmatrix}
    \pi_{11}^{[rs]}(t) \\
    \pi_{22}^{[rs]}(t)
  \end{bmatrix} \\
  &&+  \begin{bmatrix}
                                                q_{21}+\frac{\lambda}{N-1}  \\
                                                q_{12}+\frac{\lambda}{N-1}
                                              \end{bmatrix}
\label{eq.UIPA2}
\end{eqnarray}
Moreover
\begin{equation}
\label{eq.pi12}
\pi_{12}^{[rs]}(t)=\pi_{21}^{[rs]}(t) = \frac{1}{2}\left(1-\pi_{11}^{[rs]}(t)-\pi_{22}^{[rs]}(t)\right)
\end{equation}
\end{thm}

This result is remarkable because, as in the case of the individual marginal distribution, also the bivariate distribution of two agents' opinions can be propagated without requiring higher order joint distributions involving three or more agents.

\medskip
\noindent {\bf Steady-state analysis}
\medskip

First recall that, in view of Theorem \ref{th.gattosenzatopo}, the steady state expectation of the frequency of opinion 1 is
\begin{equation}
\label{eq.UIPAmean}
E[n_1/N]=q_{21}/(q_{12}+q_{21})=\tilde \pi_1
\end{equation}
for any influence intensity $\lambda \ge 0$ and any network size $N$.
This is interesting as it shows that the expectation of opinion frequencies coincides with that of the non-interacting case, irrespective of the influence intensities, when unbiased.
However, the value of $\lambda$ affects the variance, as shown in the next theorem that addresses the steady state situation. The proof can be found in the Appendix.

\begin{thm}
\label{th.var}
For the UIPA network of size $N$, the steady state variance of the frequency of opinion 1 for given influence intensity $\lambda \ge 0$ is
\begin{equation}
\label{eq.var}
{\rm Var}[n_1/N]=\frac{\tilde \sigma^2}{N} \left(1 + \frac{\lambda (N-1)}{\lambda + (q_{12}+q_{21})(N-1)}   \right)
\end{equation}
where
\begin{equation}
\label{eq.varisol}
\tilde \sigma^2 = \frac{q_{12} q_{21}}{(q_{12}+q_{21})^2}
\end{equation}
\end{thm}

The formula \eqref{eq.var} of Theorem \ref{th.var} reveals that, for a fixed $\lambda$, the variance of the frequency $n_1/N$ tends to zero as the size $N$ of the network goes to $\infty$. Hence, the ergodic distribution of $n_1/N$ becomes deterministic for large networks. This is a consequence of the completeness of the interaction graph and the fact that the total influence on each agent linearly depends on the fraction of neighbors' opinions, and this fraction converges to its expectation in view of the ergodic law of large numbers.

As for the case of a network with fixed size $N$, the variance \eqref{eq.var} is an increasing function of the influence intensity $\lambda$. It varies monotonically from the variance $\tilde \sigma^2/N$ of the non-interacting case to
a value coinciding with the variance $\tilde \sigma^2$ of the opinion of a single isolated agent. The asymptotic value with $\lambda \to \infty$ reflects a sort of synchronization of the network producing a herd behaviour, with all agents changing their opinions unanimously as one.
Rather interestingly, the herding phenomenon was observed in previous studies on social learning, by using quite different models and assumptions, see e.g. \cite{Bowden08}, \cite{Gaio02}.

\subsection{Biased influence}
When $\lambda_1 \ne \lambda_2$ there is an asymmetry in the influence intensity. If $\lambda_1 > \lambda_2$, each agent exerts more influence on its neighbors when its opinion is 1. It is expected that, compared to the non-interacting network, the distribution is affected in favor of the more influential opinion, increasing the average of $n_1(t)$.
An analytical study of the effect of $\lambda_1$ and $\lambda_2$ on the opinion dynamics becomes more difficult since the marginalization results cannot be applied. Nevertheless, the transient can be studied from the numerical solution of \eqref{eq.trans}, while the steady-state moments can be computed from the steady-state distribution specified by \eqref{eq.pi}, \eqref{eq.p0}.

\begin{figure}[tbh]
\centering
\epsfig{file=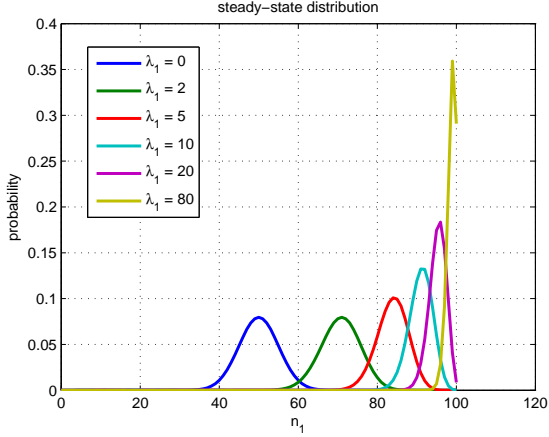, height=6cm}
\caption{Effect of the unilateral promotion on the steady-state distribution of $n_1$ for a PA composed of $N=100$ agents with $q_{12}=q_{21}=1$. The curves correspond to different values of the intensity parameter $\lambda_1$, while $\lambda_2$ is kept equal to zero.}
\label{Fig.BIPA_dist}
\end{figure}

\begin{figure}[tbh]
\centering
\epsfig{file=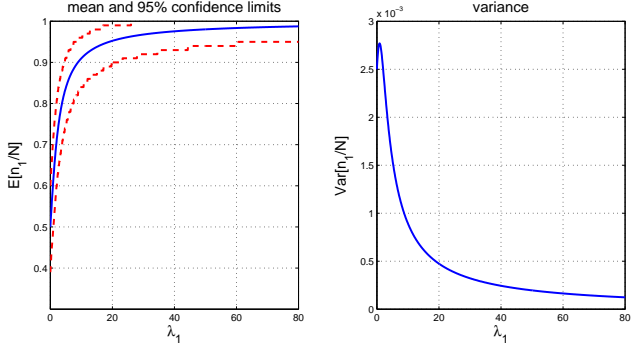, height=4.7cm}
\caption{Effect of the unilateral promotion on the mean and the variance of $n_1/N$ for a PA composed of $N=100$ agents with $q_{12}=q_{21}=1$ for different values of $\lambda_1$. The red dashed curves in the left panel correspond to the 2.5 and 97.5 percentiles.}
\label{Fig.BIPA_mv}
\end{figure}

A special important case occurs when one the two intensity parameters is zero, say $\lambda_2=0$. In this unilateral promotion case, the social influence of opinion 1 can be enhanced through the tuning knob $\lambda_1$, while agents with opinion 2 do not interact in the network. It is expected that increasing $\lambda_1$ pushes the probability distribution of $n_1$ to the right, with a limit deterministic distribution concentrated in $n_1=N$, that corresponds to an unanimous consensus on opinion 1. This is confirmed by the plots of the steady-state distribution computed for a PA of $N=100$ agents with $q_{12}=q_{21}=1$ and different values of $\lambda_1$ from eqs. \eqref{eq.pi}, \eqref{eq.p0}, reported in Figure \ref{Fig.BIPA_dist}.
The mean and variance of $n_1/N$ against $\lambda_1$ are shown in Figure \ref{Fig.BIPA_mv}. It can be noticed that the mean rises rapidly with $\lambda_1$ and approaches 1 asymptotically, while the profile of the variance is not monotonic and decays to zero more slowly.

\begin{figure}[th]
\centering
\epsfig{file=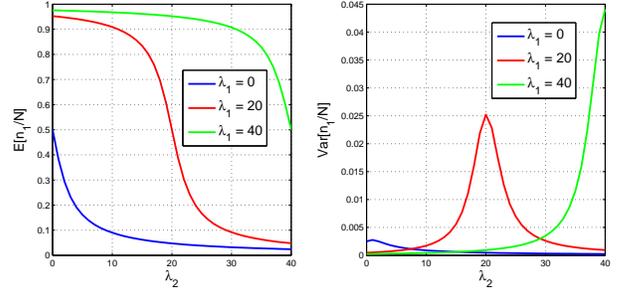, height=4.1cm}
\caption{Mean and variance of $n_1/N$ for a PA of $N=100$ agents with $q_{12}=q_{21}=1$, as functions of $\lambda_2$, when $\lambda_1$ assumes three fixed values.}
\label{Fig.BIPA_mvbis}
\end{figure}

The situation becomes more complex when both the intensity parameters $\lambda_1$ and $\lambda_2$ are different from zero. In particular, Figure \ref{Fig.BIPA_mvbis} displays the mean and variance of $n_1/N$ for a PA of $N=100$ agents with $q_{12}=q_{21}=1$, as functions of $\lambda_2$, when $\lambda_1$ assumes three fixed values. A remarkable feature is the peak of the variance when the two intensities are close to each other.

\section{Simulation examples}
\label{examples}
In this section we present and discuss several simulations in order to illustrate the previous theoretical results as well as explore opinion behaviors under assumptions for which analytical results are not yet available

\subsection{Peer Assembly with unbiased influence}
In order to demonstrate some properties of the PA, we carried out simulations with $N=100$ agents and unbiased influence with different values of $\lambda$, namely $\lambda=0$ (noninteracting network), $\lambda=2$ and $\lambda=10$.
The entries of the stand-alone transition rate matrix of each agent are $q_{12}=q_{21}=1$, corresponding to steady-state probability $\tilde \pi_1 = \tilde \pi_2 = 0.5$ and variance $\tilde \sigma^2 = 0.25$.

The theoretical steady-state values for $E[n_1/N]$ and $Var[n_1/N]$ given in \eqref{eq.UIPAmean} and \eqref{eq.var} are reported in Table \ref{Tab.UIPA}. Recall that the value of $\lambda$ only affects the variance. Conversely, the mean, both in the transient and in steady-state is not influenced by the value of $\lambda$, see Proposition \ref{th.UIPAmean} and Theorem \ref{th.var}.

\begin{table}[h]
  \centering
\begin{tabular}{r|ccc|}
&$\lambda=0$&$\lambda=2$&$\lambda=10$ \\ \hline
$E[n_1/N]$&0.5&0.5&0.5\\ \hline
$Var[n_1/N]$&0.0025&0.0050&0.0144\\ \hline
\end{tabular}
 \caption{Steady-state mean and variance of an UIPA model with $N=100$, $q_{12}=q_{21}=1$, and different values of $\lambda$.}
 \label{Tab.UIPA}
\end{table}

The simulation were performed using the birth-death chain starting from three different initial distributions $p(0)$: (i) binomial, (ii) uniform and (iii) deterministic. The binomial distribution corresponds to the probability of having $k$ agents among $N$ in opinion 1 under the mutual independence assumption. The uniform distribution assumes that the probability of $k$ agents having opinion 1 is constant for all $k$. The deterministic distribution assumes that no agent has opinion 1 with probability 1.

By combining the three initial distributions and the three values of $\lambda$, 9 scenarios were simulated, drawing 5 realizations of $n_1(t)/N$ for each one. The results are displayed in Figure \ref{Fig.UIPA_sim1}, where the realizations are plotted along with the theoretical mean and the 2.5 and 97.5 percentiles. In each row, the initial distribution is the same and the value of $\lambda$ varies. While the time profile of the mean is unchanged, the width of the $95\%$ interval gets larger as the intensity influence increases, leading to increased variability of the sample paths. For a fixed value of $\lambda$ each column shows the different transient behaviour caused by the three initializations. In view of ergodicity, the steady-state distribution is always the same, which reflects on the steady-state mean and percentiles.

\begin{figure}[th]
\epsfig{file=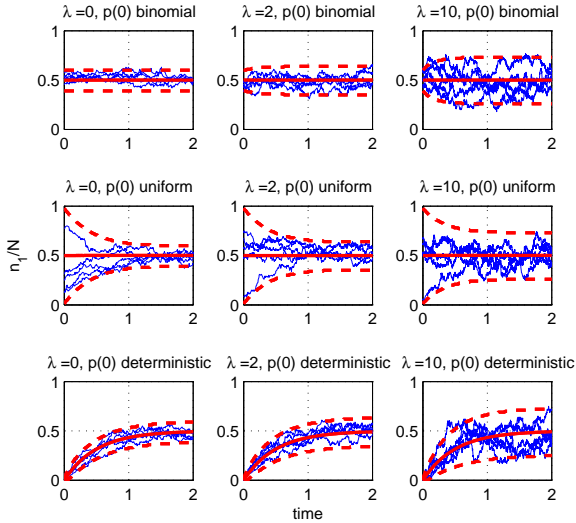, height=7.8cm}
\caption{Simulations of the UIPA model with $N=100$, $q_{12}=q_{21}=1$, different values of $\lambda$ and different initial distributions. For each scenario 5 realizations are displayed. The red lines indicate the theoretical time profile of the mean value (solid) and the 2.5 and 97.5 percentiles (dashed).}
\label{Fig.UIPA_sim1}
\end{figure}

A second set of simulations was performed to illustrate the herd behaviour of the social network when the parameter $\lambda$ is large. For an UIPA of $N=20$ agents, three scenarios are depicted in Figure \ref{Fig.UIPA_herd}, showing the effect of different values of $\lambda$, namely $\lambda=10,20,200$. It is assumed that the initial distribution $p(0)$ is binomial and coincides with the steady-state distribution in the non-interacting case. The red curves represent the theoretical mean and the $2.5$ and $97.5$ percentiles. A single realization is plotted for each $\lambda$.
The three lower panels display the steady-state distribution $\bar p_i$ of the number of agents sharing opinion 1.

For the smallest $\lambda$ (left panel), the steady-state distribution, though different from the binomial, is still unimodal. As $\lambda$ increases this distribution becomes almost uniform (middle panel) and eventually converges towards a bimodal distribution concentrated in the extreme points (right panel). As predicted by \eqref{eq.var}, the steady-state variance is monotonically increasing with $\lambda$. This is also reflected in the increased variability of the realizations. In particular, the realization with $\lambda=200$ exhibits a herd behaviour, with significant dwelling times in the two extreme situations where the agents are unanimous.

\begin{figure}[th]
\centering
\epsfig{file=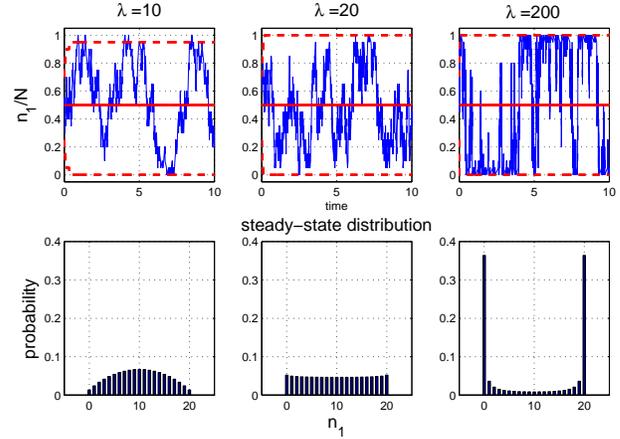, height=6cm}
\caption{Simulations of the UIPA model with $N=20$, $q_{12}=q_{21}=1$, different values of $\lambda$ and initial binomial distribution. For each scenario a single realization is displayed in the upper panels. The red lines indicate the theoretical time profile of the mean value (solid) and the 2.5 and 97.5 percentiles (dashed). The lower panels show the corresponding steady-state distribution of the number of agents sharing opinion 1. For $\lambda=200$, the emergence of the herd behaviour can be observed.}
\label{Fig.UIPA_herd}
\end{figure}

\subsection{Peer Assembly with biased influence}
We now consider the case of a PA with biased influence.
In particular, we simulated the effect of a gradual stepwise increase of $\lambda_2$ from 0 to 40 when $\lambda_1$ is kept constant and equal to 20, see Figure \ref{Fig.BIPA_oprev}. The increase of $\lambda_2$ produces a majority reversal from opinion 1 to opinion 2. Remarkably, this entails an intermediate turbulent phase where the variance is significantly higher than at start or arrival.

\begin{figure}[th]
\centering
\epsfig{file=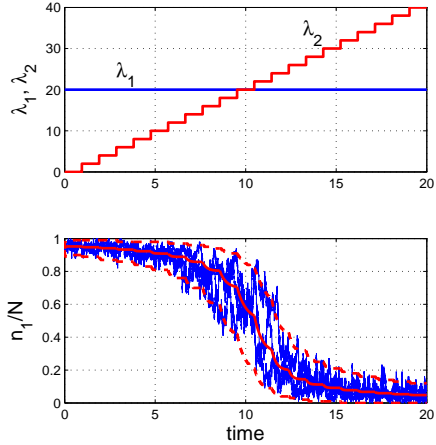, height=6.5cm}
\caption{Simulations of the PA model with $N=100$ agents, $q_{12}=q_{21}=1$, $\lambda_1=20$ and a stepwise pattern of the intensity parameter $\lambda_2$. Three realizations are displayed along with the time profiles of the mean (red solid) and the 2.5 and 97.5 percentiles (red dashed).}
\label{Fig.BIPA_oprev}
\end{figure}

In order to appreciate the dynamic effect of stepwise changes of both influence intensities, a further simulation experiment was designed according to the following setup. The total time interval $[0,10]$ is partitioned in four segments $\mathcal T_1 =[0 \quad 1]$, $\mathcal T_2 =[1\quad 4]$, $\mathcal T_3 =[4\quad 7]$, $\mathcal T_4 =[7\quad 10]$. As for the intensities, it was assumed that
$\lambda_1=\lambda_2=0$, for $t\in \mathcal T_1$;
$\lambda_1=20, \lambda_2=0$, for $t\in \mathcal T_2$;
$\lambda_1=\lambda_2=20$, for $t\in \mathcal T_3$;
$\lambda_1=16, \lambda_2=20$, for $t\in \mathcal T_4$.
In other words, starting from a non interacting network during $\mathcal T_1$, the social network is first subject to a unilateral promotion in favor of opinion 1 during $\mathcal T_2$. Then, by switching $\lambda_2$ to the same value as $\lambda_1$, an UIPA configuration is maintained during $\mathcal T_3$. Finally, in $\mathcal T_4$ the influence is biased in favor of opinion 2.
Figure \ref{Fig.BIPA_step} displays the theoretical mean of $n_1/N$ (red) along with the $95\%$ confidence limits (red dashed) and 3 Monte Carlo realizations.
It is worth noting the changes of the mean as well as the width of the confidence band. In particular, in $\mathcal T_1$ and $\mathcal T_3$ the steady-state mean is the same (since $\lambda_1=\lambda_2$ in both intervals) but the variance is much larger when the intensities are nonzero. It is also remarkable that in $\mathcal T_4$, a ratio $\lambda_1/\lambda_2=0.8$ definitely moves the average in favor of opinion 2.

\begin{figure}[th]
\centering
\epsfig{file=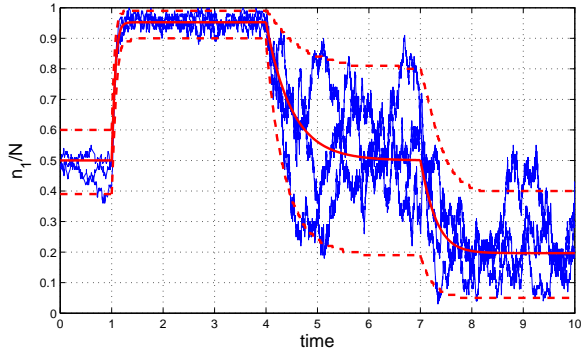, height=5cm}
\caption{Simulations of the PA model with $N=100$ agents, $q_{12}=q_{21}=1$, and a stepwise pattern of the intensity parameters $\lambda_1$ and $\lambda_2$. Three realizations are displayed along with the time profiles of the mean (red solid) and the 2.5 and 97.5 percentiles (red dashed).}
\label{Fig.BIPA_step}
\end{figure}

\subsection{Network with general topology}
The previous simulations regarded the special case of binary opinions and a complete interaction graph.
While the extension to more than two opinions could be worked out through multi-dimensional birth and death chains,
substantial analytical difficulties arise when the topology of the social network departs from the complete graph.
Nevertheless, one can still resort to the Markov Master model of Section \ref{MasterMarkovModel} and carry out Monte Carlo simulation studies.

For illustrative purposes, we discuss the effect of different network topologies for a two-opinion model with $N=100$ individuals, all having stand-alone transition rate matrix defined by $q_{12}=q_{21}=1$. As for the influence, both the unbiased and biased cases are considered.

The considered topologies are: (a) non-interacting, (b) complete (Peer Assembly), (c) small-world, (d) star. The smallworld topology is obtained according to the model introduced in \cite{Strogatz98} letting $k=1$ and $p=0.2$. In the star topology, a central agent is connected to $N-1$ peripheral agents, which do not communicate with each other.

The unbiased influence case with $\lambda_1=\lambda_2=10$ is considered first.
In Figure \ref{Fig.multitopo-u}, for each topology we report a single Monte Carlo simulation of $n_1/N$ (left column), and the estimate of the steady-state distribution of $n_1/N$ computed from 10 replications of the Monte Carlo simulation (right column).
Moreover the sample estimates of the steady-state mean and variance of $n_1/N$ are reported in Table \ref{Tab.multitopo-u}.

\begin{table}[h]
  \centering
\begin{tabular}{r|cccc|}
&(a)&(b)&(c)&(d) \\ \hline
$E[n_1/N]$&0.4978&0.4943&0.5111&0.5053\\ \hline
$Var[n_1/N]$&0.0025&0.0135&0.0099&0.1322 \\ \hline
\end{tabular}
  \caption{Estimated steady-state mean and variance for different topologies and unbiased influence.}
  \label{Tab.multitopo-u}
\end{table}

\begin{figure}[th]
\centering
\epsfig{file=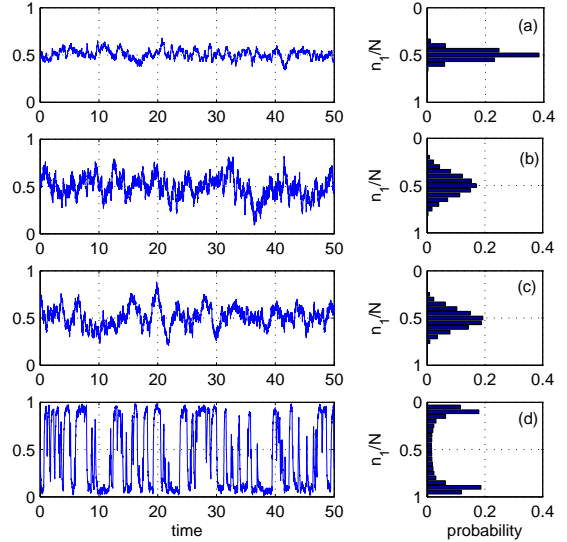, height=8cm}
\caption{Simulations and steady state distributions of $n_1/N$ with $N=100$ agents, $q_{12}=q_{21}=1$, $\lambda_1=\lambda_2=10$ and different topologies: (a) noninteracting, (b) complete, (c) smallworld, (d) star.}
\label{Fig.multitopo-u}
\end{figure}

In accordance with Theorem \ref{th.gattosenzatopo}, it appears that the expectation of $n_1/N$ is not affected by the topology: in fact all the sample means are close to 0.5. The steady-state variance instead depends on the topology. It is the smallest in the non-interacting case. Indeed, the distribution of $n_1$ is binomial and $Var[n_1]/N^2 = \tilde\sigma^2/N = 0.0025$. When the graph is complete the sample variance increases, yielding a value in good agreement with the theoretical value derived from eq. \eqref{eq.var} for the UIPA model ($Var[n_1/N]=0.012$). The decreased connectivity of the smallworld topology explains why the sample variance of case (c) is slightly smaller. Finally, the star topology triggers large opinions waves giving rise to a bimodal distribution of $n_1/N$, that justifies the large value of the sample variance in case (d).

Figure \ref{Fig.multitopo-b1}  and Table \ref{Tab.multitopo-b1} correspond to the biased case with $\lambda_1=1$ and $\lambda_2=0$, i.e. unilateral promotion. For the sake of comparison, the non-interacting case (a) is also displayed in Figure \ref{Fig.multitopo-b1}.

From the estimates of the mean, the possible effect of the topology on the expectation of $n_1/N$ is hardly appreciated. Whether or not this is due to a general topology-free invariance property is an open question that would deserve further investigation. On the contrary, the topology affects the variance. In particular, for the star topology the variance is larger and the steady-state distribution is flatter.

\begin{table}[h]
  \centering
\begin{tabular}{r|ccc|}
&(b)&(c)&(d) \\ \hline
$E[n_1/N]$&0.6167&0.6085&0.6152\\ \hline
$Var[n_1/N]$&0.0024&0.0026&0.0060 \\ \hline
\end{tabular}
  \caption{Estimated steady-state mean and variance for different topologies and biased influence.}
  \label{Tab.multitopo-b1}
\end{table}

\begin{figure}[th]
\centering
\epsfig{file=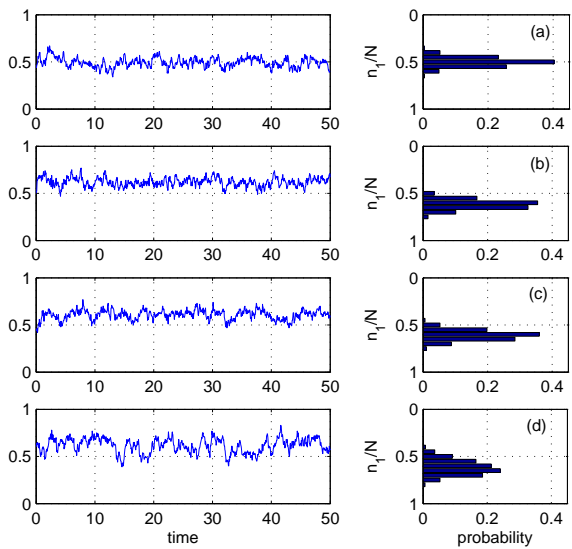, height=8cm}
\caption{Simulations and steady state distributions of $n_1/N$ with $N=100$ agents, $q_{12}=q_{21}=1$, $\lambda_1=1, \lambda_2=0$ and different topologies: (a) noninteracting, (b) complete, (c) smallworld, (d) star.}
\label{Fig.multitopo-b1}
\end{figure}

\section{Discussion and concluding remarks}
\label{Conclusions}
A major contribution of the present work is the proposal of a stochastic multi-agent model for opinion dynamics that explicitly accounts for a centralized tuning of the strength of interaction between individuals within a social network.
The aim is to gain insight on the effects of filtering algorithms managed by social network platforms.

The proposed model, consisting of Markovian agents, is very flexible and can be easily adapted to describe a variety of situations, including different topologies and the presence of heterogeneous agents. In this respect, it is important to note that the overall model preserves the Markov chain structure, thus being viable for Monte Carlo simulation.

There are special cases that can be studied analytically in order to highlight the emergence of particular behaviours. One such case is the so-called Peer Assembly that assumes binary opinions, identical agents, and a complete graph.
Thanks to lumpability into a birth-death chain, we have been able to obtain closed formulas for the evolution of the opinions distribution.

A key marginalization result (Theorems \ref{th.marginalization} and \ref{th.gattosenzatopo}) has been worked out for  arbitrary graph topology and number of opinions, provided that the agents are indistinguishable and the influence is unbiased.
Through marginalization, it has been demonstrated that the influence strength does not affect the expected number of agents sharing a certain opinion, but does affect the variance, which is associated with opinion fluctuations. This might suggest tuning strategies of the influence parameters in order to artificially freeze or excite the variability of opinions.
For instance, imagine that Facebook tunes the News Feed algorithm during a political election campaign so as to decrease the probability that a generic user is exposed to friends' posts with political content (whatever its orientation).
This is equivalent to decreasing the value of the parameter $\lambda$ in an unbiased context. Our model predicts no change in the mean number of individuals sharing a certain opinion, but a decrease of the variance, which prevents the occurrence of large random deviations from the mean. Hence, the decreased exposure to political posts would play in favor of the leading party, whose supremacy would be more hardly challenged.

Conversely, a biased influence eases the spread of some opinions to the detriment of the others. For instance, algorithmic curbing of fake news might produce such an effect. Filtering the news according to the authoritativeness of the source could spoil opinions supported by independent blogs and nongovernmental organizations against those broadcasted by mainstream media.

A possible interesting research development is the analysis and design of feedback control policies aimed at driving the collective opinion to a desired target. Although centralized control may evoke worrisome scenarios, this kind of intervention could also be used to break ``filter bubbles" \cite{Bozdag15} in order to foster diversity of perspectives among the users.

\section*{Appendix}
{\bf Proof of Theorem \ref{th.UIPAjoint}}
\begin{proof}
Consider a generic pair $(r,s)$ of agents.
First of all, note that \eqref{eq.pi12} immediately follows from the observation that the two agents are indistinguishable as a consequence of the assumption on the initial probabilities and the joint probabilities sum up to one.

Now, recall that the symbol $\Sigma(t)$ stands for the state of the Master Markov model. Note that all states $\Sigma(t)$ having $k$ entries equal to 1, $0\le k\le N$, share the same probability, here denoted by $\rho_k(t)$. By standard combinatorial calculus, it follows that ${\rm Pr}\{n_1(t)=k\}=\rho_k(t)\binom{N}{k}$. So:
\begin{eqnarray*}
  \pi_{11}^{[rs]}(t+dt) &=& {\rm Pr}\{\sigma^{[r]}(t+dt)=1, \sigma^{[s]}(t+dt)=1\} \\
  &&\hspace{-5mm}=\sum_{k=2}^N \left(1-2dt\left(q_{12}+\lambda \frac{N-k}{N-1} \right) \right)\rho_k(t)\binom{N-2}{k-2} \\
  &&\hspace{-5mm} + 2 \sum_{k=2}^N dt\left(q_{21}+\lambda \frac{k-1}{N-1} \right) \rho_{k-1}(t)\binom{N-2}{k-2}
\end{eqnarray*}
Hence
\begin{eqnarray*}
  \pi_{11}^{[rs]}(t+dt) &=& \sum_{k=2}^N \rho_k(t)\binom{N-2}{k-2}\\
   &&+ 2dt\sum_{k=2}^N \left(q_{21}+\lambda \frac{k-1}{N-1} \right)\rho_{k-1}(t)\binom{N-2}{k-2} \\
   &&- 2dt \sum_{k=2}^N \left(q_{12}+\lambda \frac{N-k}{N-1} \right) \rho_k(t)\binom{N-2}{k-2}
\end{eqnarray*}
Observing that
\begin{eqnarray*}
  \sum_{k=2}^N \rho_k(t)\binom{N-2}{k-2} &=& \pi_{11}^{[rs]}(t) \\
  \sum_{k=2}^N \rho_{k-1}(t)\binom{N-2}{k-2} &=& \pi_{12}^{[rs]}(t)
\end{eqnarray*}
we obtain
\begin{eqnarray}
\nonumber  \pi_{11}^{[rs]}(t+dt) &=& \pi_{11}^{[rs]}(t)+ 2dt \left( - q_{12}\pi_{11}^{[rs]}(t) + q_{21}\pi_{12}^{[rs]}(t)\right) \\
   &&\hspace{-1.8cm} + \frac{2\lambda dt}{N-1}
   \sum_{k=2}^N \left((k-1) \rho_{k-1}(t) - (N-k) \rho_k(t)\right)\binom{N-2}{k-2}
   \nonumber \label{ep.pi11}
\end{eqnarray}
Note that
\begin{eqnarray}
&&\sum_{k=2}^N \left((k-1) \rho_{k-1}(t) - (N-k) \rho_k(t)\right)\binom{N-2}{k-2}\nonumber\\
&& \hspace{1cm}= \pi_{12}^{[rs]}(t) - (N-2) \pi_{11}^{[rs]}(t)\nonumber\\
&& \hspace{1.5cm}+ \sum_{k=2}^N (k-2)(\rho_{k-1}(t)+\rho_k(t)) \binom{N-2}{k-2}
\label{eq.proofap1}
\end{eqnarray}
Now, consider the term
$$
\eta(t) = \sum_{k=2}^N (k-2)(\rho_{k-1}(t)+\rho_k(t)) \binom{N-2}{k-2}
$$
By observing that
$$
(k+1)\binom{N-2}{k+1} = (N-2)\binom{N-3}{k}
$$
it follows that
\begin{eqnarray*}
  \eta(t) &=& \sum_{k=0}^{N-3} (k+1)(\rho_{k+2}(t)+\rho_{k+3}(t)) \binom{N-2}{k+1} \\
  &=& (N-2) \sum_{k=0}^{N-3} (\rho_{k+2}(t)+\rho_{k+3}(t)) \binom{N-3}{k} \\
  &=& (N-2) \sum_{k=0}^{N-3}\rho_{k+2}(t)\binom{N-3}{k} \\
  &&+ (N-2) \sum_{k=1}^{N-2}\rho_{k+2}(t)\binom{N-3}{k-1} \\
  &=& (N-2) \left( \rho_{2}(t) + \rho_{N}(t)\right) \\
  &&+ (N-2) \sum_{k=1}^{N-3}\rho_{k+2}(t)\left(\binom{N-3}{k}+\binom{N-3}{k-1}\right)
\end{eqnarray*}
Noting that
$$
\binom{N-3}{k}+\binom{N-3}{k-1} = \binom{N-2}{k}
$$
the previous expression becomes
\begin{eqnarray*}
  \eta(t) &=& (N-2) \left( \rho_{2}(t) + \rho_{N}(t) + \sum_{k=1}^{N-3}\rho_{k+2}(t)\binom{N-2}{k} \right) \\
          &=& (N-2) \sum_{k=0}^{N-2}\rho_{k+2}(t)\binom{N-2}{k} \\
          &=& (N-2) \sum_{k=2}^{N}\rho_{k}(t)\binom{N-2}{k-2} = (N-2) \pi_{11}^{[rs]}(t)
\end{eqnarray*}
Then, recalling \eqref{eq.proofap1}, \eqref{eq.pi12} and taking the limit for $dt$ tending to zero, the first row of eq. \eqref{eq.UIPA2} is demonstrated.
As for the second row, it follows from symmetry by simply exchanging the indices.
\end{proof}

\medskip
{\bf Proof of Theorem \ref{th.var}}
\begin{proof}
First of all note that
\begin{eqnarray}
\nonumber
{\rm Var}[n_1]&=&E\left[\left(\sum_{r=1}^N \mathcal I_{\sigma^{[r]}=1}]\right)^2\right] - \left(E[n_1] \right)^2\\
&=& N\tilde\pi_1 + N(N-1) \tilde \pi_{11} - N^2 \tilde\pi_1^2
\label{eq.varn}
\end{eqnarray}
where $\tilde\pi_1$ is given by \eqref{eq.UIPAmean} and $\tilde \pi_{11}$ is the steady-state limit of $\pi_{11}(t)$ satisfying eq. \eqref{eq.UIPA2} (which exists in view of ergodicity). From \eqref{eq.UIPA2} at the equilibrium we obtain that
$$
\tilde \pi_{11} = \begin{bmatrix}
  1 & 0
\end{bmatrix}(D + b c^\prime)^{-1}b
$$
with
$$
D=2\begin{bmatrix}
                    q_{12} & 0 \\
                    0 & q_{21}
                  \end{bmatrix}, \quad
b=\begin{bmatrix}
                                                q_{21}+\frac{\lambda}{N-1}  \\
                                                q_{12}+\frac{\lambda}{N-1}
                                              \end{bmatrix}, \quad
c=\begin{bmatrix}                              1 \\ 1
                                                            \end{bmatrix}
$$
By observing that $(D + b c^\prime)^{-1}b = D^{-1}b/(1+c^\prime D^{-1}b)$, it results that
\begin{eqnarray*}
  \tilde \pi_{11}&=& \frac{1}{(1+c^\prime D^{-1}b)}\begin{bmatrix}
                       1 & 0
                     \end{bmatrix} D^{-1}b\\
  &=& \frac{q_{21}\left(\lambda + q_{21}(N-1)\right)}{(q_{12}+q_{21})\left(\lambda + (q_{12}+q_{21})(N-1)\right)}
\end{eqnarray*}
Now, by replacing $\tilde \pi_1$ and $\tilde \pi_{11}$ in \eqref{eq.varn} and dividing by $N^2$, the result directly follows.
\end{proof}

\end{document}